\documentclass[letterpaper,10pt,conference]{ieeeconf}

\IEEEoverridecommandlockouts
\usepackage{amsmath}
\usepackage{amssymb}
\usepackage{amsfonts}
\usepackage{mathtools}

\usepackage{graphicx}
\usepackage{booktabs}
\usepackage{balance}
\usepackage{multirow}

\usepackage{algorithmic}

\usepackage{cite}
\usepackage{xcolor}

\newtheorem{theorem}{Theorem}

\newtheorem{assumption}{Assumption}

\newcommand{\hmin}{h_{\min}}
\newcommand{\Ss}{\mathcal{S}_{\mathrm{s}}}
\newcommand{\D}{\mathcal{D}}
\newcommand{\Hs}{\mathcal{H}}
\newcommand{\X}{\mathcal{X}}
\newcommand{\Ad}{A_\delta}
\newcommand{\Ts}{T_{\mathrm{s}}}

\title{\LARGE \bf
Exact-Safe MPPI: Safety-Aware Sampling with Nonsmooth Control Barrier Functions
}

\author{Alexander Coker and Leilei Cui%
\thanks{A. Coker and L. Cui are with the Department of Mechanical Engineering, University of New Mexico.}%
}
\usepackage{amsmath,amssymb}
\usepackage{tikz}
\usetikzlibrary{positioning,arrows.meta,fit,calc}
\usepackage{url}

\usepackage[ruled,vlined,linesnumbered]{algorithm2e}
\usepackage{float}
\DontPrintSemicolon

\definecolor{mppiorange}{RGB}{252,224,190}
\definecolor{qpblue}{RGB}{206,226,246}
\definecolor{subblue}{RGB}{232,241,251}
\definecolor{plantgy}{RGB}{223,228,235}

\begin{document}

\maketitle
\thispagestyle{empty}
\pagestyle{empty}

\begin{abstract}
Model predictive path integral (MPPI) control plans actions by sampling and evaluating trajectories, but standard MPPI does not guarantee safety. Incorporating control barrier function (CBF) filters into each rollout allows the planner to account for safety corrections. However, combining multiple CBFs through a soft-minimum approximation can exclude states that satisfy the original safety constraints, restricting the trajectories available to the planner. This paper proposes Exact-Safe MPPI, which uses the exact minimum of multiple barrier functions and applies a CBF-based quadratic program during both planning and control execution. To handle nondifferentiability where multiple barriers attain the minimum, the filter simultaneously enforces the CBF conditions for all $\delta$-active barriers, whose values lie within a prescribed tolerance $\delta>0$ of the minimum. Under suitable feasibility and regularity assumptions, we establish forward invariance of the certified safe set without the shrinkage introduced by soft-minimum smoothing. We also characterize the state-dependent approximation gap and explain how multiple barriers with values at or near the minimum increase this shrinkage. Furthermore, we show that increasing the number of MPPI samples cannot recover trajectories excluded by the soft-minimum filter. We demonstrate the effectiveness of the proposed approach in a navigation task requiring passage through a narrow gap between two obstacles. Video demonstrations are available at \url{https://github.com/lc-lab25/Exact-Safe-MPPI}.
\end{abstract}


\section{Introduction}
Model predictive control (MPC) is an optimization-based feedback control approach that repeatedly plans over a finite horizon. At each time step, it uses the current system state to compute a control sequence that minimizes a specified cost while satisfying the system dynamics and constraints, including safety requirements. Only the first control input is applied, and the process repeats using the updated state. Initially adopted in industrial process control, MPC has since found widespread use in robotics, transportation, and other fields~\cite{rawlings2017model,borrelli2017predictive}. However, solving the underlying optimization problem can be computationally demanding, particularly for nonlinear systems with nonsmooth or discontinuous costs and complex constraints. Building on path integral methods for stochastic optimal control~\cite{Kappen_2005}, model predictive path integral (MPPI) control offers a sampling-based approach within the same receding-horizon framework~\cite{williams17,williams17b}. At each planning step, MPPI samples candidate control sequences, simulates the resulting trajectories, and updates the control sequence by assigning greater weight to samples with lower costs. These simulations can run in parallel without requiring derivatives of the dynamics or cost, and can accommodate nonsmooth or discontinuous costs, making MPPI attractive for robotic applications~\cite{XueICRA2025,pan2026hybrid}. However, standard MPPI may generate many trajectories that violate the robot’s safety constraints, spending substantial computation on unsafe candidate trajectories that contribute little to the control update.

To improve the safety of MPPI, one can penalize constraint violations in the objective function~\cite{williams17} or bias sampling toward trajectories that remain within the safe region~\cite{biased}. These approaches encourage safe behavior but do not, by themselves, guarantee safety. They may also spend substantial computation evaluating unsafe trajectories that contribute little to the control update. Another approach combines MPPI with control barrier function (CBF)-based safety mechanisms~\cite{shield,tao,gsmppi}. If a safety filter modifies the proposed control but these modifications are not considered during planning, the predicted trajectories may differ from the resulting motion. This mismatch can degrade performance and lead to repeated interventions, oscillations, or ``chattering''~\cite{hu2024active,leung2020infusing,Bejarano2025,oh2026provably}. Accounting for the safety filter during planning can help reduce these effects. In~\cite{tao}, CBF conditions guide the selection of the mean and covariance of the sampled controls through a semidefinite program, which can be more computationally demanding than a standard QP-based CBF filter. GS-MPPI~\cite{gsmppi} instead applies a QP-based CBF filter to the sampled controls within each rollout. MPPI therefore evaluates trajectories generated by the filtered controls and accounts for the filter when updating its proposed control sequence. To handle multiple safety constraints, GS-MPPI combines their barrier functions into a single smooth barrier using a soft-minimum approximation. However, combining multiple CBFs through a soft-minimum approximation can exclude
trajectories in the exact certified set, causing the planner to treat an otherwise
feasible solution as infeasible, as illustrated in Section~IV.B.

In this paper, we propose Exact-Safe MPPI, a framework that incorporates multiple CBFs into MPPI using their exact minimum rather than a soft-minimum approximation (see Fig. \ref{fig:pipeline}). During planning, a CBF-QP filter modifies the sampled controls within each rollout. The same filter is applied to the control proposed by MPPI before it is executed on the system. To represent multiple safety constraints, we use the exact minimum $\hmin(x)=\min_j h_j(x)$ instead of the soft minimum $h_\rho(x)=-\tfrac{1}{\rho}\ln\sum_{j=1}^{\ell}e^{-\rho h_j(x)}$ used in~\cite{gsmppi}. Because $h_\rho(x)\leq\hmin(x)$, the exact minimum avoids the reduction in the certified safe set introduced by the soft-minimum approximation. Although $\hmin$ may be nondifferentiable where multiple barriers attain the minimum, its Clarke generalized gradient can be expressed in terms of the gradients of these active barriers (i.e. $h_j=\hmin$ for all active $j$) \cite{kamaldar}. Motivated by this characterization and following \cite{glot17,glot19,glot21,kamaldar}, we enforce the CBF conditions simultaneously for all $\delta$-active barriers, defined by $h_j(x)-\hmin(x)\leq\delta$, where $\delta>0$ is a prescribed tolerance. In summary, our contributions are twofold:
\begin{enumerate}
\item We develop Exact-Safe MPPI by incorporating an exact-minimum CBF-QP filter into both trajectory sampling and control execution. This framework accounts for safety corrections during planning while avoiding the safe-set shrinkage introduced by soft-minimum composition.
\item We characterize the state-dependent gap between the soft minimum $h_\rho$ and the exact minimum $\hmin$. This analysis explains how barriers with equal or nearly equal minimum values increase the approximation gap and lead to a nonuniform reduction in the certified safe set.
\end{enumerate}

\section{Preliminaries}

Consider the control-affine nonlinear system
\begin{equation}\label{eq:plant}
\dot x=f(x)+g(x)u,
\end{equation}
where $x\in\mathbb{R}^n$ is the state, $u\in\mathbb{R}^m$ is the control input, and $f:\mathbb{R}^n\to\mathbb{R}^n$ and $g:\mathbb{R}^n\to\mathbb{R}^{n\times m}$ are locally Lipschitz. Safety requirements of system \eqref{eq:plant} are specified by multiple smooth constraint functions $h_j:\mathbb{R}^n\to\mathbb{R}$, $j=1,\ldots,\ell$, with each constraint requiring $h_j(x)\geq0$. The individual super-level sets and their intersection define the safe set:
\begin{equation}\label{eq:original-safe-set}
\mathcal{C}_{j,0}\triangleq
\{x\in\mathbb{R}^n:h_j(x)\geq0\},
\qquad
\Ss\triangleq\bigcap_{j=1}^{\ell}\mathcal{C}_{j,0}.
\end{equation}

\subsection{Control barrier functions}

For a differentiable scalar function $h$, let $L_fh(x)\triangleq\nabla h(x)^\top f(x)$ and $L_gh(x)\triangleq\nabla h(x)^\top g(x)$ denote its Lie derivatives. Control barrier
functions (CBFs) provide sufficient conditions for rendering a prescribed safe set forward invariant~\cite{ames17}. For a relative-degree-one
constraint $h_j(x)\geq0$, the standard CBF condition is
\begin{equation}\label{eq:cbf}
L_fh_j(x)+L_gh_j(x)u+\alpha_j(h_j(x))\geq0,
\end{equation}
where $\alpha_j$ is an extended class-$\mathcal{K}$ function \cite{ames17}. Any locally Lipschitz controller satisfying \eqref{eq:cbf} on $\mathcal{C}_{j,0}$ renders that set forward
invariant under the standard CBF hypotheses.

When $h_j$ has relative degree $d_j>1$, the input $u$ does not appear in its first time derivative, and \eqref{eq:cbf} cannot directly constrain the input. High-order CBFs recursively introduce intermediate barriers until the input appears~\cite{xiao,nguyen}.

\subsection{High-Order Barrier Formulation}\label{sec:composite}

\begin{assumption}[Relative degree]\label{as:rd}
Each constraint $h_j$ has relative degree $d_j$ with respect to
\eqref{eq:plant} on $\Ss$; that is,
$L_gL_f^r h_j\equiv0$ for all $r\le d_j-2$ and
$L_gL_f^{d_j-1}h_j\neq0$.
\end{assumption}

For each constraint $h_j$, the high-order CBF recursion starts with the
original barrier $b_{j,0}=h_j$ and produces a sequence of barrier functions.
The intermediate barrier functions preserve the lower-order safety requirements,
while the final barrier function $b_{j,d_j-1}$ is the first whose condition
depends on the control input. Following~\cite{xiao,nguyen}, define
\begin{equation}\label{eq:casc}
b_{j,i+1}(x)
=
L_f b_{j,i}(x)+\alpha_{j,i}(b_{j,i}(x)),
\qquad
0\le i\le d_j-2,
\end{equation}
where $\alpha_{j,i}\in C^{d_j-1-i}$ are extended class-$\mathcal K$
functions. For each recursion level, define
\[
\mathcal{C}_{j,i}
\triangleq
\{x\in\mathbb{R}^n:b_{j,i}(x)\ge0\}.
\]
The set of states satisfying all lower-order conditions is
\begin{equation}\label{eq:C}
\mathcal{C}
\triangleq
\bigcap_j\bigcap_{i\le d_j-2}\mathcal{C}_{j,i}.
\end{equation}

For the final barrier condition, we use one common locally Lipschitz extended
class-$\mathcal K$ function $\alpha$ across all constraints; that is,
$\alpha_{j,d_j-1}\equiv\alpha$ for every $j$. The lower-order recursion
functions $\alpha_{j,i}$, $i\le d_j-2$, may remain constraint-specific. If
the control input $u$ satisfies
\begin{align}\label{eq:highCBF}
    &L_f b_{j,d_j-1}(x)
    +L_g b_{j,d_j-1}(x)u
    +\alpha\bigl(b_{j,d_j-1}(x)\bigr)
    \geq0 \nonumber\\
    &\hspace{35mm} j=1,\ldots,\ell,
\end{align}
for all $x\in\Ss$, then the exact high-order certified set
\begin{equation}\label{eq:exact-set}
\D
\triangleq
\mathcal{C}
\cap
\bigcap_{j=1}^{\ell}\mathcal{C}_{j,d_j-1}
\subseteq\Ss
\end{equation}
is forward invariant~\cite{xiao,nguyen}.

GS-MPPI~\cite{gsmppi} combines the final barrier functions into one
differentiable scalar barrier using the soft minimum
\begin{equation}\label{eq:comp}
h_\rho(x)
\triangleq
-\tfrac1\rho
\ln\textstyle\sum_{j=1}^{\ell}
e^{-\rho b_{j,d_j-1}(x)},
\end{equation}
with $\rho>0$. This is a smooth approximation of the generally
nondifferentiable exact minimum
\begin{equation}\label{eq:exact-min}
\hmin(x)
\triangleq
\min_j b_{j,d_j-1}(x).
\end{equation}
The soft minimum is a smooth under-approximation~\cite{nesterov}, satisfying
$h_\rho(x)\leq\hmin(x)$ for all $x\in\mathbb{R}^n$.

Combining the lower-order conditions in $\mathcal{C}$ with the soft-minimum
condition gives the soft-minimum certified set
\begin{equation}\label{eq:sets}
\Hs
\triangleq
\{x\in\mathbb{R}^n:h_\rho(x)\ge0\}
\cap\mathcal{C}.
\end{equation}
The intersection with $\mathcal C$ is necessary because the final barrier
conditions alone do not imply the lower-order intermediate barrier conditions
when $d_j>1$. Since $\{\hmin\ge0\}\cap\mathcal C=\D$ and
$h_\rho\le\hmin$, $\Hs\subseteq\D$. The high-order CBF recursion itself may
give $\D\subsetneq\Ss$; all exactness claims are relative to $\D$.

\subsection{Model Predictive Path Integral (MPPI)}

MPPI is a sampling-based receding-horizon control method that evaluates sampled control sequences through forward simulation and updates the sampling distribution using exponential cost weights~\cite{williams17,williams17b}. Consider the discrete-time prediction model $x_{k+1}=F(x_k,v_k)$ obtained by discretizing \eqref{eq:plant}. Let $N\in\mathbb N$ denote the planning horizon and $M=(\mu_0,\ldots,\mu_{N-1})$ the mean desired-control sequence. For the $s${th} rollout, let $x_0^{(s)}$ be the current state and let $\epsilon_k^{(s)}\sim\mathcal{N}(0,\Sigma)$ be independent perturbations with covariance $\Sigma\succ0$. MPPI generates $K$ candidate sequences and trajectories by sampling
\begin{equation}\label{eq:mppi_rollouts}
v_k^{(s)}=\mu_k+\epsilon_k^{(s)},\qquad
x_{k+1}^{(s)}=F(x_k^{(s)},v_k^{(s)}),
\end{equation}
for $k=0,\ldots,N-1$. Let $V=(v_0,\ldots,v_{N-1})$ denote the random control sequence, and let $\mathbb P_M$ denote its probability measure with the mean sequence $M$. The $s$th sampled realization is denoted by $V^{(s)}=(v_0^{(s)},\ldots,v_{N-1}^{(s)})$. For a standard importance-sampling formulation, the rollout score is
\begin{equation}\label{eq:mppi_cost}
J^{(s)}
=\phi(x_N^{(s)})
+\sum_{k=0}^{N-1}\left[
q(x_k^{(s)})
+\lambda \mu_k^\top\Sigma^{-1}\epsilon_k^{(s)}
\right],
\end{equation}
where $q$ and $\phi$ denote the stage state cost and the terminal cost, respectively, and $\lambda>0$ is the temperature parameter. The cross term accounts for sampling around
the mean sequence under the standard coupling $R=\lambda\Sigma^{-1}$ between the
quadratic control-cost matrix and the perturbation covariance; terms common to all
rollouts are omitted. The mean sequence is updated according to
\begin{align}\label{eq:mppi_update}
\begin{split}
\mu_k^+ &=\mu_k+\sum_{s=1}^{K}w^{(s)}\epsilon_k^{(s)}, \\
w^{(s)}
&=\frac{\exp\!\left[-(J^{(s)}-J_{\min})/\lambda\right]}
{\sum_{r=1}^{K}\exp\!\left[-(J^{(r)}-J_{\min})/\lambda\right]},
\end{split}
\end{align}
where $J_{\min}:=\min_s J^{(s)}$ improves numerical stability without changing the
weights. The first updated input $\mu_0^+$ is passed to the safety filter.

GS-MPPI uses a differentiable soft minimum to combine multiple barrier functions into
one scalar barrier. Figure~\ref{fig:prelim-gap} illustrates the resulting certified-set
shrinkage for a two-constraint passage. The solid black curves bound the exact
high-order certified set $\D$, whereas the red dashed curves bound the smaller
soft-minimum certified set $\Hs$. Along the dotted tie axis, the two component barriers
have equal value. The shaded region $\D\setminus\Hs$ satisfies every component safety
constraint but is excluded by soft-minimum smoothing; in this example, the excluded
region closes the feasible passage. Thus, near states where multiple barriers are
simultaneously restrictive, smoothing can make a feasible solution unavailable to the
planner. This paper avoids that unnecessary loss by using the exact barrier composition
within MPPI rollouts.
\begin{figure}[t]
    \centering
    \includegraphics[width=\columnwidth]{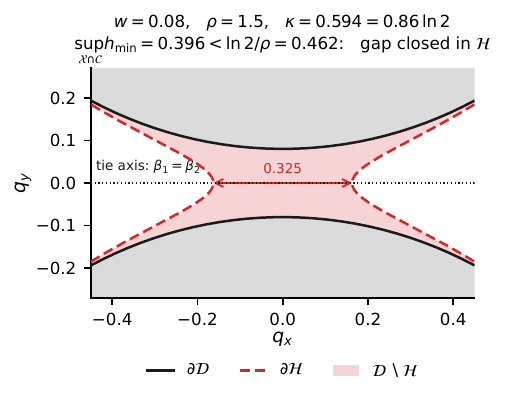}
    \caption{The soft-minimum certified set $\Hs$ is a subset of the exact high-order certified set $\D$. Near simultaneous barrier constraints, soft-minimum smoothing can exclude states that satisfy all component constraints.}
    \label{fig:prelim-gap}
\end{figure}

\section{Exact constraints composition}\label{sec:exact}
Since $h_\rho$ is a soft-minimum approximation of $\hmin$ satisfying $h_\rho(x)\leq\hmin(x)$, the soft-minimum construction certifies the smaller set $\Hs\subseteq\D$. This inclusion can be strict, introducing conservatism in the certified safe set. To address this conservatism, we work directly with the exact minimum $\hmin$, which may be nondifferentiable at points where multiple component barriers attain the minimum. At such points, its Clarke generalized gradient is the convex hull of the gradients of the active component barriers~\cite{clarke}. Enforcing a condition for only one active barrier does not, in general, ensure that the other active constraints are satisfied~\cite{glot17}. We therefore develop a construction that simultaneously enforces all $\delta$-active constraints. We refer to the resulting method as \emph{Exact-Safe MPPI}.

\subsection{The safety filter with $\delta$-active barriers}\label{sec:filter}

At a tie where multiple component barriers attain the minimum of $\hmin(x)$, more than one component barrier can be relevant to safety. Rather than selecting one component barrier, Exact-Safe MPPI enforces all barriers whose values are close to the exact minimum. Fix an active-set tolerance $\delta>0$ and define
\begin{equation}\label{eq:active-set}
\Ad(x)\triangleq\left\{j\in\{1,\ldots,\ell\}:
b_{j,d_j-1}(x)-\hmin(x)\leq\delta\right\}.
\end{equation}
Thus, $\Ad(x)$ contains every component barrier lying within $\delta$ of the minimum, and $k_\delta(x)\triangleq|\Ad(x)|$ is the number of $\delta$-active barriers. For each $j\in\Ad(x)$, define
\begin{equation}\label{eq:component-cbf}
\begin{aligned}
a_j(x)&\triangleq L_f b_{j,d_j-1}(x)
+\alpha(b_{j,d_j-1}(x)),\\
B_j(x)&\triangleq L_g b_{j,d_j-1}(x)
\in\mathbb{R}^{1\times m}.
\end{aligned}
\end{equation}
Let $a(x) \in \mathbb{R}^{k_{\delta}}$ and $B(x)  \in \mathbb{R}^{k_{\delta} \times m}$ denote the vector and matrix obtained by vertically stacking $a_j$ and $B_j$ over $j \in \Ad(x)$, respectively.

A larger $\gamma$ penalizes slack variables $\hat{\mu} = [\hat{\mu}_1,\cdots,\hat{\mu}_{k_\delta}]^\top$ more heavily and therefore favors correcting the control input over relaxing a constraint away from its boundary. The exact-safe QP filter is formulated as
\begin{equation}\label{eq:qp}
\begin{aligned}
\min_{\hat u,\hat\mu}\quad
&\tfrac12\|\hat u-v\|^2+\tfrac\gamma2\|\hat\mu\|^2 \\
\mathrm{s.t.}\quad
&a_j(x)+B_j(x)\hat u+\hat\mu_j b_{j,d_j-1}(x)\geq0,
\qquad j\in\Ad(x).
\end{aligned}
\end{equation}
Following the multiplicative slack construction in~\cite{gsmppi}, $\hat\mu_j$ can relax a component constraint away from its highest-order boundary. Its effect vanishes when $b_{j,d_j-1}(x)=0$, so the unrelaxed CBF condition must hold at the boundary of $\mathcal{C}_{j,d_{j-1}}$. Consequently, the filter enforces every $\delta$-active component barrier as it approaches its boundary.

At a simultaneous boundary, the QP is feasible exactly when the corresponding component-barrier conditions admit a common feasible control input. This is the standard compatibility requirement for any controller that enforces multiple CBFs simultaneously. 

\begin{figure}[t]
\centering
\begin{tikzpicture}[
  x=1cm,
  y=1cm,
  font=\footnotesize,
  block/.style={
    draw=black!75,
    rounded corners=2pt,
    align=center,
    inner xsep=4pt,
    inner ysep=3pt,
    line width=0.5pt
  },
  ar/.style={->, line width=0.55pt}
]

\node[
  block,
  fill=green!17!black!4,
  minimum width=2.10cm
] (planner) at (1.35,3.20) {%
  MPPI Planner\\[-1pt]
  {\scriptsize(11)--(13)}};

\node[
  draw=red!45!blue,
  dashed,
  rounded corners=3pt,
  fill=red!8!blue!10!white,
  minimum width=3.35cm,
  minimum height=3.35cm,
  inner sep=0pt,
  line width=0.5pt
] (filterbox) at (4.50,2.60) {};

\node[
  anchor=north west,
  text=red!45!blue,
  font=\footnotesize\bfseries
] at (2.96,4.12) {Exact Safety Filter};

\node[
  block,
  fill=cyan!22!white,
  minimum width=2.35cm
] (qp) at (4.50,3.20) {%
  Exact Safety QP\\[-1pt]
  {\scriptsize(16)}};

\node[
  block,
  fill=yellow!28!white,
  minimum width=2.55cm
] (select) at (4.50,1.70) {%
    $\delta$-Active Barrier\\[-1pt]
    Selection\\[-1pt]
  {\scriptsize$\mathcal A_\delta(x)$ \;(14)}};

\node[
  block,
  fill=red!18!white,
  minimum width=1.80cm
] (system) at (7.55,3.20) {%
  System\\[-1pt]
  {\scriptsize(1)}};

\draw[ar] (planner.east) --
  node[pos=0.22,above=2pt] {$v$} (qp.west);

\draw[ar] (select.north) --
  node[pos=0.50,right=4pt,align=center]
  {$a_j,\,B_j$\\[-1pt]{\scriptsize(15)}} (qp.south);

\draw[ar] (qp.east) --
  node[pos=0.78,above=2pt] {$\hat u$} (system.west);

\coordinate (xlowR) at (7.55,0.62);
\coordinate (xbranch) at (4.50,0.62);
\coordinate (xlowL) at (1.35,0.62);

\draw[ar] (system.south) -- (xlowR)
          -- node[below=2pt] {$x$} (xlowL)
          -- (planner.south);

\draw[ar] (xbranch) -- (select.south);

\end{tikzpicture}
\caption{Exact-Safe MPPI. The system state $x$ is used by both the MPPI planner
and $\delta$-active barrier selection. The planner proposes $v$, the QP computes the
filtered input $\hat u$, and the system receives $\hat u$.}
\label{fig:pipeline}
\end{figure}
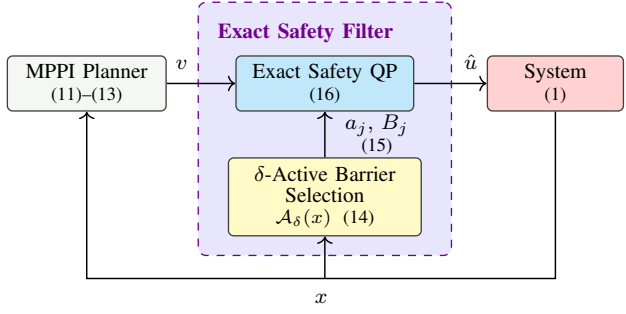
\subsection{Forward invariance under simultaneous enforcement}

For each component barrier, enforcing the highest-order CBF condition in \eqref{eq:highCBF} preserves the lower-order intermediate barrier inequalities, provided they hold initially. Thus, enforcing these conditions for all component barriers renders $\D$, defined in \eqref{eq:exact-set}, forward invariant~\cite{xiao,nguyen}. This result applies separately to each component barrier, including when the components have different relative degrees. It therefore remains only to show that the QP enforces the required highest-order condition whenever a component barrier reaches its boundary. Because the $\delta$-active set can change discontinuously at tie states, closed-loop solutions are interpreted in the Filippov sense~\cite{cortes}.

\begin{theorem}[Forward invariance of $\D$]\label{thm:inv}
Under Assumption~\ref{as:rd}, let $\delta>0$, and suppose that the QP \eqref{eq:qp} is feasible in a neighborhood of $\D$ and returns a measurable, locally bounded control $\hat u$ with locally bounded slack vector $\hat\mu$. Then every Filippov solution of
\[
\dot x=f(x)+g(x)\hat u(x)
\]
starting in $\D$ remains in $\D$ throughout its interval of existence.
\end{theorem}

\begin{proof}
Consider a highest-order component barrier $b_{j,d_j-1}$ at a state $\bar x\in\D$ for which $b_{j,d_j-1}(\bar x)=0$. Since every highest-order barrier is nonnegative in $\D$, $\hmin(\bar x)=0$, and therefore $j\in\Ad(\bar x)$. By continuity of $b_{j,d_j-1}$ and $\delta>0$, the component $b_{j,d_j-1}$ remains $\delta$-active in a neighborhood of $\bar x$.

At the boundary, the slack term in \eqref{eq:qp} vanishes. Thus, the QP constraint gives the unrelaxed component CBF condition
\[
L_fb_{j,d_j-1}
+
L_gb_{j,d_j-1}\hat u
+
\alpha(b_{j,d_j-1})
\geq0.
\]
Because $\hat\mu$ is locally bounded, the slack contribution $\hat\mu_jb_{j,d_j-1}$ vanishes as $b_{j,d_j-1}\to0$. The resulting boundary inequality is preserved under the Filippov convexification of the closed-loop dynamics. The standard comparison argument~\cite{cortes,khalil} therefore prevents $b_{j,d_j-1}$ from becoming negative.

The standard high-order CBF recursion then implies that nonnegativity of every final barrier function preserves the lower-order intermediate barrier functions defining $\D$~\cite{xiao,nguyen}. Hence every condition defining $\D$ remains satisfied, and the trajectory stays in $\D$.
\end{proof}

\subsection{Integration with MPPI rollouts}

Within each rollout, the sampled control becomes the desired input $v$, and the QP
filter computes the filtered input $\hat u$ applied to the system
\eqref{eq:plant}. The proposed Exact-Safe MPPI is shown in Alg. \ref{alg:exact-safe-mppi}. The MPPI update \eqref{eq:mppi_update} then reweights the resulting
filtered trajectories. The planner uses one Euler step of duration $\Ts$ per rollout
step~\cite{gsmppi}; during execution, the filter is re-evaluated at the smaller substep $\delta t$.

Under exact integration, Theorem~\ref{thm:inv} shows that every feasible filtered rollout beginning in $\D$ remains in $\D$. Thus, the rollout distribution is supported on the reachable set of the filtered closed loop. Sections~\ref{sec:cost}
and~\ref{sec:disc} analyze how soft-minimum composition restricts that set and what
finite-step Euler rollouts inherit.

\begin{algorithm}
\caption{Exact-Safe MPPI Planner}
\label{alg:exact-safe-mppi}
\KwIn{$x_0,\ M,\ \Sigma,\ N,\ K,\ \delta,\ \gamma$}
$M=(\mu_0,\ldots,\mu_{N-1})$: previous best control sequence\;
\KwOut{Filtered control $\hat u_0$}

\For{$s\gets1$ \KwTo $K$}{
    $x_0^{(s)}\gets x_0$\;
    \For{$k\gets0$ \KwTo $N-1$}{
        $\epsilon_k^{(s)}\sim\mathcal N(0,\Sigma)$\;
        $v_k^{(s)}\gets\mu_k+\epsilon_k^{(s)}$\;
        $\mathcal A_\delta^{(s)}\gets\mathcal A_\delta(x_k^{(s)})$\;
        $\hat u_k^{(s)}\gets
        \operatorname{QP}_{\eqref{eq:qp}}(x_k^{(s)},v_k^{(s)})$\;
        $x_{k+1}^{(s)}\gets F(x_k^{(s)},\hat u_k^{(s)})$\;
    }
    $J^{(s)}\gets J(\tau^{(s)})$\;
}
Compute $w^{(s)}$ from \eqref{eq:mppi_update} for $s=1,\ldots,K$\;
\For{$k\gets0$ \KwTo $N-1$}{
    $\mu_k^+\gets\mu_k+\sum_{s=1}^{K}w^{(s)}\epsilon_k^{(s)}$\;
}
$\hat u_0\gets\operatorname{QP}_{\eqref{eq:qp}}(x_0,\mu_0^+)$\;
\Return $\hat u_0$\;
\end{algorithm}

\section{Effect of barrier composition on sampling-based planning}\label{sec:cost}

The standard pointwise approximation bound for \eqref{eq:comp} is
$0\le\hmin-h_\rho\le\ln\ell/\rho$~\cite{rh_softmin,molnar}.
The following refinement shows that the soft minimum differs appreciably from the
exact minimum only where multiple component barriers have nearly equal minimum values.

\begin{theorem}[Tie-localized certified-set shrinkage]\label{thm:erosion}
Fix $x$, let $A(x)\triangleq\{j:b_{j,d_j-1}(x)=\hmin(x)\}$, $k\triangleq|A(x)|\ge1$,
and $\Gamma(x)\triangleq\min_{j\notin A(x)}(b_{j,d_j-1}(x)-\hmin(x))>0$, with
$\Gamma\triangleq+\infty$ if $k=\ell$. Then for all $\rho>0$,
\begin{equation}\label{eq:identity}
\hmin-h_\rho=\frac{\ln k}{\rho}
+\frac1\rho\ln\Big(1+\frac1k\sum_{j\notin A(x)}e^{-\rho(b_{j,d_j-1}-\hmin)}\Big),
\end{equation}
and consequently
\begin{equation}\label{eq:sandwich}
\frac{\ln k}{\rho}\ \le\ \hmin-h_\rho\ \le\ \frac{\ln k}{\rho}+\frac{\ell-k}{k\rho}e^{-\rho\Gamma(x)}.
\end{equation}
\end{theorem}

\begin{proof}
Factoring $e^{-\rho\hmin}$ from the soft-minimum sum separates the $k$ tying
barriers from the remaining barriers. From \eqref{eq:comp},
\begin{equation}\label{eq:erosion-factorization}
\hmin-h_\rho=\frac1\rho\ln\!\left(
k+\sum_{j\notin A(x)}e^{-\rho(b_{j,d_j-1}-\hmin)}\right),
\end{equation}
which gives \eqref{eq:identity} after factoring out $k$. Each untied barrier lies at
least $\Gamma(x)$ above the minimum, so its exponential contribution is at most
$e^{-\rho\Gamma(x)}$. Positivity gives the lower bound in \eqref{eq:sandwich}, while
$\ln(1+s)\le s$ gives the upper bound and the stated tie and non-tie behavior.
\end{proof}

It is observed from Theorem \ref{thm:erosion} that on a tie set ($k\ge2$), the approximation gap decreases proportionally to
$1/\rho$, with leading term $\ln k/\rho$; away from the tie sets ($k=1$),
$0<\hmin-h_\rho\le(\ell-1)e^{-\rho\Gamma(x)}/\rho$.

\subsection{Soft-minimum limits exploration and performance}
In GS-MPPI~\cite{gsmppi}, the soft-minimum safety filter is applied throughout every rollout, limiting which states the planner can explore. With exact integration, every rollout starting in $\Hs$ remains in $\Hs$. Since $\Hs\subseteq\D$, these rollouts cannot reach states in $\D\setminus\Hs$, even though those states belong to the original certified safe set $\D$. Increasing the number of samples only
explores the trajectories that remain in $\Hs$ more thoroughly; it cannot recover a trajectory that requires passing through $\D\setminus\Hs$. By comparison, under the assumptions of Theorem~\ref{thm:inv}, Exact-Safe MPPI rollouts remain in $\D$ without the additional restriction imposed by soft-minimum smoothing.

This restriction can impose a performance gap that persists regardless of the number of samples. Suppose the lowest-cost trajectory that remains within the
exact certified set $\D$ has cost $J^\star$ and must pass through
$\D\setminus\Hs$. If every trajectory that remains within $\Hs$ costs at least
$\Delta>0$ more than this trajectory, then any trajectory $\tau$ generated or
executed by soft-minimum GS-MPPI satisfies
\begin{equation}\label{eq:irreducible-cost-gap}
J(\tau)\geq J^\star+\Delta.
\end{equation}
The soft-minimum filter prevents the planner from considering the lower-cost
trajectory because it passes through $\D\setminus\Hs$. Increasing the number of
samples cannot remove this gap because it does not enlarge the set of trajectories
that the filter permits.

\subsection{Gap closing in a two-obstacle passage}\label{sec:threshold}

Figure~\ref{fig:gap-closing-example} illustrates a vehicle passing from the left side of two obstacles to the right side through the gap between them. Every such trajectory must cross the vertical slice
\[
\X\triangleq\{(q,\dot q):q_x=0\}
\]
where $q = [q_x,q_y]^\top$ denotes the position of the vehicle, and $\dot q=[\dot q_x,\dot q_y]^\top$ is the velocity. At the center of this slice, the upper and lower obstacle barriers are equally
restrictive. The exact minimum retains a passage through $\X$, whereas the soft
minimum is lower by $\ln2/\rho$ at this two-barrier tie. As shown in the figure, this
loss can close the soft-minimum passage even though every component barrier remains
feasible there.

\begin{figure}[t]
    \centering
    \includegraphics[width=\columnwidth]{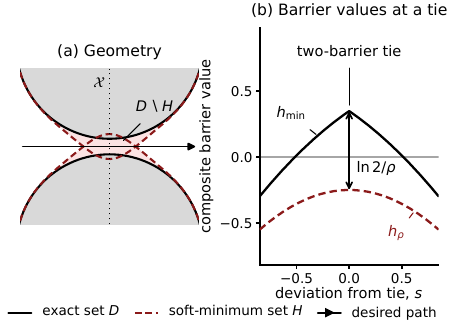}
    \caption{Two-barrier gap closing. \textbf{Left:} every left-to-right trajectory
    must cross the slice $\X$. The exact certified set $\D$ retains the passage,
    whereas the soft-minimum set $\Hs$ can close it; the shaded region is
    $\D\setminus\Hs$. \textbf{Right:} at a two-barrier tie, $h_\rho$ lies
    $\ln2/\rho$ below $h_{\min}$. When this loss exceeds the available barrier margin
    on $\X$, the soft-minimum certified set does not reach the crossing slice.}
    \label{fig:gap-closing-example}
\end{figure}

We evaluate this mechanism on rounded-square and circular-obstacle geometries. The
two-disk calculation below gives the result for circular obstacles.

\emph{Two-disk calculation.}
Consider two disk obstacles with radius $r>0$ and a gap of width
$2w$, where $w>0$ is the half-gap width. Their centers are
\[
o_{1,2}=(0,\pm(r+w)).
\]
The corresponding barrier functions are $h_j(q)=\|q-o_j\|^2-r^2$, and the high-order CBF gain is
$\alpha_{j,0}(s)=as$ for $a>0$.

For this symmetric geometry, the largest exact-barrier value on the crossing slice
occurs at the center of the gap, where $q_y=\dot q_y=0$. Direct evaluation gives
\[
\sup_{\X\cap\mathcal C}\hmin=a\,w(2r+w)>0.
\]
Thus, the exact certified set always reaches the crossing slice for $w>0$. At the
same point, the soft minimum is lower by $\ln2/\rho$. Therefore, the soft-minimum
certified set closes the crossing slice exactly when
\begin{equation}\label{eq:disk-gap-closure}
\rho\,a\,w(2r+w)<\ln2.
\end{equation}
Any inactive speed constraint can only further reduce the soft minimum; it cannot
reopen a passage that the obstacle barriers have closed.

For the two-disk example, define the normalized quantity
\begin{equation}\label{eq:kappa-threshold}
\kappa\triangleq\rho\,a\,w(2r+w)
=
\rho\sup_{\X\cap\mathcal C}\hmin,
\qquad
\kappa_{\mathrm{crit}}=\ln2.
\end{equation}
More generally, Theorem~\ref{thm:erosion} gives the leading value $\ln k$ at a
$k$-constraint tie. Additional constraints can only increase the certified-set
shrinkage, so the observed transition can occur above, but not below, $\ln2$.
Together with the support restriction above, this excludes soft-minimum
continuous-time rollouts from threading the gap below the threshold. The value
$\ln2$ is derived rather than fitted to the experiments.

\section{Experiments}\label{sec:expt}

\subsection{Gap-closing experiment}\label{sec:gap-expt}

The theory predicts that soft-minimum smoothing can close a gap that remains available
under the exact barrier composition. We test whether the predicted threshold
$\kappa_{\mathrm{crit}}=\ln2$ identifies when this occurs in practice.

We use a shared double-integrator vehicle model with two obstacle geometries:
rounded-square obstacles represented by $p$-norm barriers ($p=8$), and circular
obstacles represented by squared-distance barriers. The vehicle must pass between the
obstacles to reach its goal. We vary the gap width and the soft-minimum parameter
$\rho$, then compare soft-minimum GS-MPPI~\cite{gsmppi} with Exact-Safe MPPI. Each
condition uses $50$ random seeds. In total, the experiment contains $4400$
soft-minimum runs and $550$ exact-composition runs. 

We call a run \emph{threaded} when it crosses the required slice $\X$ between the
obstacles. A run \emph{completes} when it also reaches the goal. We exclude the
narrowest disk gap because the exact-composition planner crossed the gap in $60\%$ of
runs but never reached the goal; that case tests goal reachability rather than
soft-minimum gap closure.

None of the $1050$ soft-minimum runs with $\kappa<\ln2$ threaded the gap
(Fig.~\ref{fig:collapse}). The smallest successful value was
$0.875=1.26\ln2$. For the disk with $w=0.12$, the exact-barrier margin available at
the crossing is $0.6060$. At $\rho=1$, the soft-minimum loss is
$\ln2/\rho=0.6931$, which is larger than this available margin. The best
soft-minimum value on the crossing slice is therefore $-0.0871$, and the gap is
closed. At $\rho=2$, the corresponding value is positive and the gap is open.

The theory gives the continuous-time gap-closing threshold exactly:
$\kappa_{\mathrm{crit}}=\ln2$. The zero-success point immediately above $\ln2$ does not contradict the theoretical
threshold. At this condition, the soft-minimum certified set reaches the crossing
slice, so the gap is geometrically open. However, the opening is sufficiently narrow
that none of the $50$ sampled rollouts found a trajectory through it.

Because we test only eight values of $\rho$, the experiments cannot determine the
exact value at which the finite-step planner begins successfully threading each gap.
As predicted, no run threaded below $\ln2$; vehicle dynamics, Euler integration, and
a narrow passage can require a larger value before successful crossing occurs. The
speed constraint makes the soft minimum slightly more conservative, so it can only
shift successful crossing above the predicted two-obstacle threshold.

\begin{figure}[t]
\centering
\includegraphics[width=\columnwidth]{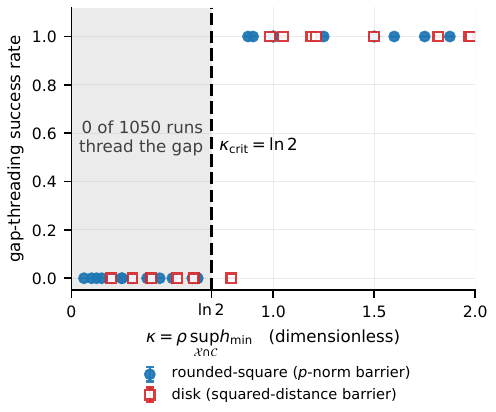}
\caption{Gap-threading success for the $4400$ soft-minimum runs. No run threaded below
the derived threshold $\ln2$. The zero-success point immediately above $\ln2$
corresponds to a geometrically open passage that was too narrow for any of the $50$
sampled rollouts to find.}
\label{fig:collapse}
\end{figure}

A closed soft-minimum gap prevents the planner from finding a path through the gap; it
does not require the vehicle to collide. The four below-threshold runs in
Table~\ref{tab:deadlock} that are neither threaded nor deadlocked remain safe and
moving when the planning horizon ends. Across all $4950$ gap-closing runs, no safety
violation was recorded; the smallest recorded value of $\min_j h_j$ was
$+2.4\times10^{-2}$ (Fig.~\ref{fig:stall}).

\begin{figure}[t]
\centering
\includegraphics[width=\columnwidth]{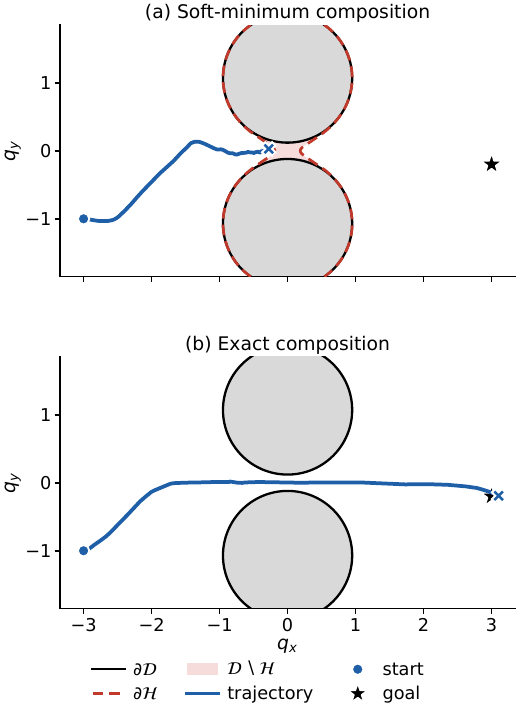}
\caption{Logged median runs at $w=0.12$. The soft minimum at $\rho=1$
($\kappa=0.606<\ln2$) stops before the gap, whereas the exact composition threads the
gap.}
\label{fig:stall}
\end{figure}

\begin{table}[t]
\caption{Soft-minimum outcomes by $\kappa$ band. Rows sum to one; moving runs count as
nonthreading.}
\label{tab:deadlock}
\centering\footnotesize
\setlength{\tabcolsep}{4pt}
\begin{tabular}{@{}lrrrrr@{}}
\toprule
& runs & threaded & deadlock & moving & viol.\\
\midrule
$\kappa<\ln2$ & 1050 & 0.000 & 0.996 & 0.004 & \textbf{0}\\
$\ln2\le\kappa<1$ & 350 & 0.857 & 0.143 & 0.000 & \textbf{0}\\
$\kappa\ge1$ & 3000 & 1.000 & 0.000 & 0.000 & \textbf{0}\\
\bottomrule
\end{tabular}
\end{table}

\subsection{Nine-constraint comparison}\label{sec:nine-expt}

The gap-closing experiment involves only two obstacle barriers. We next test whether
simultaneously enforcing all $\delta$-active barriers matters in a more complicated map
with nine constraints.

The GS-MPPI paper~\cite{gsmppi} does not specify the obstacle geometry used in its
benchmark. We therefore reconstruct an approximate version from the published figure
using rounded-square obstacles. This is a controlled reconstruction, not an exact
replication of the original benchmark.

The published temperature $\lambda=1$ makes the first MPPI update in our
reconstruction depend almost entirely on one rollout. We therefore use
$\lambda=300$ for every controller in this comparison so that the update uses many
rollouts. All controller settings use the same costs, random seeds, integration method,
and step sizes.

We compare soft-minimum GS-MPPI over several values of $\rho$ with Exact-Safe MPPI.
We also test a single-constraint version of the exact filter that enforces only the
most restrictive barrier at each state, rather than all $\delta$-active barriers together.
Each controller setting is tested from four goals using $100$ random seeds per goal.

Exact-Safe MPPI produced no recorded safety violations in the nine-constraint map.
Soft-minimum GS-MPPI produced no violations at
$\rho\in\{5,20,50\}$, followed by $34$ violations at $\rho=200$ and $53$ violations
at $\rho=1000$. These large-$\rho$ failures are numerical: although $h_\rho$ remains
differentiable for finite $\rho$, it becomes difficult to evaluate accurately near
barrier ties as $\rho$ becomes very large.

The single-constraint filter produced safety violations and greater control chattering.
In contrast, simultaneously enforcing all $\delta$-active component barriers produced no
violations in the tested runs. This supports the need for simultaneous enforcement in
the proposed filter.
\section{Discussion}\label{sec:disc}


The safety result in Theorem~\ref{thm:inv} applies to the continuous-time filtered system and does not directly extend to its Euler discretization. With a sufficiently large integration step, a numerical trajectory may leave the certified safe set even when the corresponding continuous-time trajectory remains safe~\cite{breeden}. In our experiments, refining the execution integration step
reduced these errors, and the executed Exact-Safe MPPI trajectories at
$\delta t=2\times10^{-3}\,\mathrm{s}$ did not leave $\D$. The guarantee remains
conditional on the QP being feasible and returning locally bounded control and slack
values.

The support argument shows that soft-minimum filtering can exclude trajectories that
pass through $\D\setminus\Hs$. When the lowest-cost exact-safe trajectory requires
such a region, the soft-minimum planner can have an unavoidable cost gap. This result
does not show that Exact-Safe MPPI finds the globally lowest-cost trajectory in $\D$,
nor that it recovers every trajectory reachable under the exact filter.

The reconstructed nine-constraint benchmark contains no goal in
$\D\setminus\Hs$. The rounded-square and disk experiments therefore isolate the
gap-closing mechanism, but they do not establish how often comparable certified-set
shrinkage affects other environments, cost functions, or sampling distributions.

\section{Conclusion}
This paper presented Exact-Safe MPPI, which incorporates an exact-minimum CBF-QP filter into both trajectory sampling and control execution of MPPI. By simultaneously enforcing the conditions for all $\delta$-active barriers, the framework accommodates multiple high-order CBFs and preserves the certified safe set under the stated feasibility and regularity assumptions. We characterized the state-dependent gap introduced by soft-minimum smoothing and showed how it can exclude safe passages from every rollout. Increasing the number of samples cannot recover these excluded paths. Navigation experiments supported this analysis: no soft-minimum rollout threaded the gap below the predicted threshold, while exact composition allowed passage without recorded safety violations in the tested cases. A nine-constraint comparison further demonstrated the importance of enforcing multiple active barriers simultaneously. Future work will address safety guarantees for sampled-data implementations.

\section*{Acknowledgment}
ChatGPT was used to improve grammar and the formal tone of the writing. The authors reviewed all edits and take full responsibility for the content of this paper.

\balance
\bibliographystyle{IEEEtran}
\bibliography{references}

\end{document}